\documentclass[11pt]{article}

\usepackage{titlesec}
\usepackage{latexsym}
\usepackage{amsmath}
\usepackage{amssymb}
\usepackage{amsfonts,amsthm}
\usepackage{enumitem}
\usepackage[top=1in, bottom=1in, left=1in, right=1in]{geometry}

\newcommand{\ignore}[1]{}

\newtheorem{theorem}{Theorem}[section]
\newtheorem{lemma}[theorem]{Lemma}

\newtheorem{corollary}[theorem]{Corollary}

\newtheorem{property}[theorem]{Property}
\newtheorem{proposition}[theorem]{Proposition}

\usepackage{amsthm}
\usepackage{thmtools}
\usepackage{thm-restate}

\usepackage{url}
\usepackage{subcaption}
\usepackage{crop}
\usepackage{enumerate}
\usepackage{etoolbox}
\usepackage{float}
\usepackage{hyperref}
\usepackage{latexsym}
\usepackage{mathtools}
\usepackage{relsize}
\usepackage{amsmath}
\usepackage{amssymb}
\usepackage{amsfonts}
\usepackage{balance}
\usepackage{tabu}
\usepackage{longtable}

\usepackage{multirow}
\usepackage{lscape}
\usepackage[lined, algoruled, linesnumbered]{algorithm2e}
\usepackage{lipsum}

\let\oldReturn\Return
\renewcommand{\Return}{\State\oldReturn}

\usepackage{abstract}
\usepackage{hyperref,xcolor}
\definecolor{winered}{rgb}{0.5,0,0}

\usepackage{enumitem}

\hypersetup
{
    pdfborder={0 0 0},
    colorlinks=true,
    linkcolor={winered},
    urlcolor={winered},
    filecolor={winered},
    citecolor={winered},
    linktoc=all,
}

\setlist[description]{leftmargin=\parindent,labelindent=\parindent}

\title{Improved Linear-Time Algorithm for Computing the $4$-Edge-Connected Components of a Graph}

\author{Loukas Georgiadis$^{1}$ \and Giuseppe F. Italiano$^{2}$ \and Evangelos Kosinas$^{1}$}

\begin{document}
\maketitle

\begin{abstract}
We present an improved algorithm for computing the $4$-edge-connected components of an undirected graph in linear time.
The new algorithm uses only elementary data structures, and it is simple to describe and to implement in the pointer machine model of computation.
\end{abstract}

\footnotetext[1]{Department of Computer Science \& Engineering, University of Ioannina, Greece. E-mail: \{loukas,ekosinas\}@cs.uoi.gr.
Research at the University of Ioannina supported by the Hellenic Foundation for Research and Innovation (H.F.R.I.) under the ``First Call for H.F.R.I. Research Projects to support Faculty members and Researchers and the procurement of high-cost research equipment grant'', Project FANTA (eFficient Algorithms for NeTwork Analysis), number HFRI-FM17-431.}
\footnotetext[2]{LUISS University, Rome, Italy. E-mail: gitaliano@luiss.it. Partially supported by MIUR, the Italian Ministry for Education, University and Research, under PRIN Project AHeAD (Efficient Algorithms for HArnessing Networked Data).}

\section{Introduction}
\label{section:introduction}

Determining or testing the edge connectivity of a graph $G=(V,E)$, as well as computing notions of connected components, is a classical subject in graph theory, motivated by several application areas (see, e.g., \cite{connectivity:nagamochi-ibaraki}), that has been extensively studied since the 1970's.
An \emph{(edge) cut} of $G$ is a set of edges $S \subseteq E$ such that $G \setminus S$ is not connected.
We say that $S$ is a \emph{$k$-cut} if its cardinality is $|S|=k$.
The \emph{edge connectivity} of $G$, denoted by $\lambda(G)$, is the minimum cardinality of an edge cut of $G$.
A graph is \emph{$k$-edge-connected} if $\lambda(G) \ge k$.
A cut $S$ separates two vertices $u$ and $v$, if $u$ and $v$ lie in different connected components of $G \setminus S$.
Vertices $u$ and $v$ are $k$-edge-connected if there is no $(k-1)$-cut that separates them. By Menger's theorem~\cite{menger}, $u$ and $v$ are $k$-edge-connected if and only if there are $k$-edge-disjoint paths between $u$ and $v$.
A \emph{$k$-edge-connected component} of $G$ is a maximal set $C \subseteq V$ such that there is no $(k-1)$-edge cut in $G$ that disconnects any two vertices $u,v \in C$ (i.e., $u$ and $v$ are in the same connected component of $G \setminus S$ for any $(k-1)$-edge cut $S$).
We can define, analogously, the \emph{vertex cuts} and the \emph{$k$-vertex-connected components} of $G$.
It is known how to compute the $(k-1)$-edge cuts, $(k-1)$-vertex cuts,  $k$-edge-connected components and $k$-vertex-connected components of a graph in linear time for $k \in \{2,3\}$~\cite{GI:ECtoVC,3-connectivity:ht,NagamochiIbaraki:3CC,dfs:t,Tsin:3CC}.
The case $k=4$ has also received significant attention \cite{3cuts:Dinitz,4C:Online,KR:4C,Kanevsky:4CC}, but until very recently, none of the previous algorithms achieved linear running time.
In particular,
Kanevsky and Ramachandran~\cite{KR:4C} showed
how to test whether a graph is $4$-vertex-connected in
$O(n^2)$ time. Furthermore,
Kanevsky et al.~\cite{Kanevsky:4CC} gave an $O(m+n\alpha(m,n))$-time algorithm to compute the $4$-vertex-connected components of a $3$-vertex-connected graph, where $\alpha$ is a functional inverse of Ackermann's function~\cite{dsu:tarjan}.
Using the reduction of Galil and Italiano~\cite{GI:ECtoVC} from edge connectivity to vertex connectivity, the same bounds can be obtained for $4$-edge connectivity. Specifically, one can test whether a graph is $4$-edge-connected in $O(n^2)$ time, and one can
compute the $4$-edge-connected components of a $3$-edge-connected graph in $O(m+n\alpha(m,n))$ time.
Dinitz and Westbrook~\cite{4C:Online} presented an $O(m+n\log{n})$-time algorithm to compute the $4$-edge-connected components of a general graph $G$ (i.e., when $G$ is not necessarily $3$-edge-connected).
Nagamochi and Watanabe~\cite{ni93} gave an $O(m+k^2n^2)$-time algorithm to compute the $k$-edge-connected components of a graph $G$, for any integer $k$.

Very recently, two linear-time algorithms for computing the $4$-edge-connected components of an undirected graph were presented in \cite{Linear4ECC:ESA21,Linear4ECC:Nadara:ESA21}.
The main part in both algorithms is the computation of the $3$-edge cuts of a $3$-edge-connected graph $G$. The algorithms operate on a depth-first search (DFS) tree $T$ of $G$~\cite{dfs:t}, with start vertex $r$, and compute $3$ types of cuts $C=\{e_1,e_2,e_3\}$, depending on the number of tree edges in $C$.
We refer to a cut $C$ that consists of $t$ tree edges of $T$ as a \emph{type-$t$ cut of $G$}.
The challenging cases are when $C$ is a type-$2$ or type-$3$ cut.
Nadara et al.~\cite{Linear4ECC:Nadara:ESA21} provided an elegant way to handle type-$3$ cuts. Specifically, they showed that computing all  type-$3$ cuts can be reduced, in linear time, to computing type-$1$ and type-$2$ cuts, by contracting the edges of $G \setminus T$.
To handle type-$2$ cuts in linear time, both \cite{Linear4ECC:ESA21} and \cite{Linear4ECC:Nadara:ESA21} require the use of the static tree disjoint-set-union (DSU) data structure of Gabow and Tarjan~\cite{dsu:gt}, which is quite sophisticated and not amenable to simple implementations.
Here, we present an improved version of the algorithm of \cite{Linear4ECC:ESA21} for identifying type-$2$ cuts, so that it only uses simple data structures.
The resulting algorithm relies only on basic properties of depth-first search (DFS)~\cite{dfs:t}, and on parameters carefully defined on the structure of a DFS spanning tree (see Section~\ref{sec:dfs}). As a consequence,
it is simple to describe and to implement, and it does not require the power of the RAM model of computation, thus implying the following \emph{new results}:

\begin{theorem}
\label{theorem:3-cuts}
The $3$-edge cuts of an undirected graph can be computed in linear time on a pointer machine.
\end{theorem}

\begin{corollary}
\label{corollary:4-connected-components}
The $4$-edge-connected components of an undirected graph can be computed in linear time on a pointer machine.
\end{corollary}

\section{Depth-first search and related notions}
\label{sec:dfs}

In this section we introduce the parameters that are used in our algorithm, which are defined with respect to a depth-first search spanning tree.
Let $G=(V,E)$ be a connected undirected graph with $n$ vertices, which may have multiple edges.
Let $T$ be the spanning tree of $G$ provided by a depth-first search (DFS) of $G$~\cite{dfs:t}, with start vertex $r$.
A vertex \emph{$u$ is an ancestor of a vertex $v$} (\emph{$v$ is a descendant of $u$}) in $T$ if the tree path from $r$ to $v$ contains $u$.
Thus, we consider a vertex to be both an ancestor and a descendant of itself.
The edges in $T$ are called \emph{tree-edges}; the edges in $E \setminus T$ are called \emph{back-edges}, as their endpoints have ancestor-descendant relation in $T$.
We let $p(v)$ denote the parent of a vertex $v$ in $T$.
If $u$ is a descendant of $v$ in $T$, we denote the set of vertices of the simple tree path from $u$ to $v$ as $\mathit{T}[u,v]$. The expressions $\mathit{T}[u,v)$ and $\mathit{T}(u,v]$ have the obvious meaning (i.e., the vertex on the side of the parenthesis is excluded).
We identify vertices with their preorder number assigned during the DFS. Thus, if $v$ is an ancestor of $u$ in $T$, then $v \leq u$.
Let $T(v)$ denote the set of descendants of $v$, and let $\mathit{ND}(v)=|T(v)|$ denote the number of descendants of $v$.
Then, vertex $u$ is a descendant of $v$ (i.e., $u\in T(v)$) if and only if $v\leq u < v+\mathit{ND}(v)$~\cite{domin:tarjan}.

Whenever $(x,y)$ denotes a back-edge, we shall assume that $x$ is a descendant of $y$. We let $B(v)$ denote the set of back-edges $(x,y)$, where $x$ is a descendant of $v$ and $y$ is a proper ancestor of $v$. Thus, if we remove the tree-edge $(v,p(v))$, $T(v)$ remains connected to the rest of the graph through the back-edges in $B(v)$.
Furthermore, we have the following property:
\begin{property}(\cite{Linear4ECC:ESA21})
\label{property:Bv}
A connected graph $G$ is $2$-edge-connected if and only if $|B(v)|>0$, for every $v\neq r$. Furthermore, $G$ is $3$-edge-connected only if $|B(v)|>1$, for every $v\neq r$.
\end{property}

We let $\mathit{b\_count}(v)=|B(v)|$ denote the number of elements of $B(v)$.
Assume that $G$ is $3$-edge-connected, and let $v\neq r$ be a vertex of $G$. Then we have $\mathit{b\_count}(v)>1$, and therefore there are at least two back-edges in $B(v)$.
We define the \emph{first low point of $v$}, denoted by $\mathit{low1}(v)$, as the minimum vertex $y$ such that there exists a back-edge $(x,y)\in B(v)$.
Also, we let $\mathit{low1D}(v)$ denote $x$, i.e., a descendant of $v$ that is connected with $\mathit{low1}(v)$ via a back-edge.
(Notice that $\mathit{low1D}(v)$ is not uniquely determined.)
Furthermore, we define the \emph{second low point of $v$}, denoted by $\mathit{low2}(v)$, as the minimum vertex $y'$ such that there exists a back-edge $(x',y')\in B(v)\setminus\{(\mathit{low1D}(v),\mathit{low1}(v))\}$, and let $\mathit{low2D}(v)$ denote $x'$.
Similarly, we define the \emph{high point of $v$}, denoted by $\mathit{high}(v)$, as the maximum $y$ such that there exists a back-edge $(x,y)\in B(v)$.
We also let $\mathit{highD}(v)$ denote a descendant of $v$ that is connected with $\mathit{high}(v)$ via a back-edge.
We let $l_1(v)$ denote the smallest $y$ for which there exists a back-edge $(v,y)$, or $v$ if no such back-edge exists. (Thus, $\mathit{low1}(v)\leq l_1(v)$.)
Furthermore, we let $l_2(v)$ denote the smallest $y$ for which there exists a back-edge $(v,y)\neq (v,l_1(v))$, or $v$ is no such back-edge exists.
It is easy to compute all $l_1(v)$, $l_2(v)$, $\mathit{low1}(v)$, $\mathit{low1D}(v)$, $\mathit{low2}(v)$ and $\mathit{low2D}(v)$ during the DFS.
For the computation of all $\mathit{high}(v)$ (and $\mathit{highD}(v)$), \cite{Linear4ECC:ESA21} gave a linear-time algorithm that uses the static tree DSU data structure of Gabow and Tarjan~\cite{dsu:gt}.

In order to gather the connectivity information that is contained in the sets $B(v)$, we also have to consider the higher ends of the back-edges in $B(v)$. Thus we define the \emph{maximum point} $M(v)$ of $v$ as the maximum vertex $z$ such that $T(z)$ contains the higher ends of all back-edges in $B(v)$. In other words, $M(v)$ is the nearest common ancestor of all $x$ for which there exists a back-edge $(x,y)\in B(v)$.
(Clearly, $M(v)$ is a descendant of $v$.)
Let $m$ be a vertex and $v_1, \dotsc, v_k$ be all the vertices with $M(v_1) =\dotsc = M(v_k) = m$, sorted in decreasing order.
Observe that $v_{i+1}$ is an ancestor of $v_i$, for every $i\in\{1,\dotsc,k-1\}$, since $m$ is a common descendant of all $v_1,\dotsc,v_k$.
Then we have $M^{-1}(m)=\{v_1,\dotsc,v_k\}$, and we define $\mathit{nextM}(v_i)$ $:=$ $v_{i+1}$, for every $i\in\{1,\dotsc,k-1\}$, and $\mathit{prevM}(v_i)$ $:=$ $v_k$, for every $i\in\{2,\dotsc,k\}$. Thus, for every vertex $v$, $\mathit{nextM}(v)$ is the successor of $v$ in the decreasingly sorted list $M^{-1}(M(v))$, and $\mathit{prevM}(v)$ is the predecessor of $v$ in the decreasingly sorted list $M^{-1}(M(v))$.

Now let $v$ be a vertex and let $u_1,\dotsc,u_k$ be the children of $v$ sorted in non-decreasing order w.r.t. their $\mathit{low1}$ point. We let $c_i(v)$ be $u_i$, if $i\in\{1,\dotsc,k\}$, and $\emptyset$ if $i>k$. (Note that $c_i(v)$ is not uniquely determined, since some children of $v$ may have the same $\mathit{low1}$ point.)
Then we call $c_1(v)$ the $\mathit{low1}$ child of $v$, and $c_2(v)$ the $\mathit{low2}$ child of $v$. We let $\tilde{M}(v)$ denote the nearest common ancestor of all $x$ for which there exists a back-edge $(x,y)\in B(v)$ with $x$ a proper descendant of $M(v)$. We leave $\tilde{M}(v)$ undefined if no such proper descendant $x$ of $M(v)$ exists.
We also define $M_{low1}(v)$ as the nearest common ancestor of all $x$ for which there exists a back-edge $(x,y)\in B(v)$ with $x$ being a descendant of the $\mathit{low1}$ child of $M(v)$, and also define $M_{low2}(v)$ as the nearest common ancestor of all $x$ for which there exists a back-edge $(x,y)\in B(v)$ with $x$ a descendant of the $\mathit{low2}$ child of $M(v)$.
We leave $M_{low1}(v)$ (resp. $M_{low2}(v)$) undefined if no such proper descendant $x$ of the $\mathit{low1}$ (resp., $\mathit{low2}$) child of $M(v)$ exists.

The following list summarizes the concepts used by \cite{Linear4ECC:ESA21} that are defined on a DFS-tree, and can be computed in linear time (except $B(v)$, which we do not compute).
Refer to Figure~\ref{figure:dfs-example} for an illustration.

\definecolor{cblue}{rgb}{0.36, 0.54, 0.66}
\definecolor{lgray}{rgb}{0.77, 0.76, 0.82}
\definecolor{cgray}{rgb}{0.52, 0.52, 0.51}
\renewcommand{\labelitemi}{\textcolor{lgray}{$\blacktriangleright$}}
\begin{itemize}
\item $B(v):=\{(x,y)\in E\setminus T\mid x \mbox{ is a descendant of } v \mbox{ and } y \mbox{ is a proper ancestor of } v\}$.
\item $l_1(v):=\mathit{min}(\{y\mid \exists (v,y)\in{B(v)}\}\cup\{v\})$.
\item $l_2(v):=\mathit{min}(\{y\mid \exists (v,y)\in{B(v)\setminus\{(v,l_1(v)\}}\}\cup\{v\})$.
\item $\mathit{low1}(v):=\mathit{min}\{y\mid\exists (x,y)\in B(v)\}$.
\item $\mathit{low1D}(v):=$ a vertex $x$ such that $(x,\mathit{low1}(v))\in{B(v)}$.
\item $\mathit{low2}(v):=\mathit{min}\{y\mid\exists (x,y)\in B(v)\setminus\{(\mathit{low1D}(v),\mathit{low1}(v))\}\}$.
\item $\mathit{low2D}(v):=$ a vertex $x$ such that $(x,\mathit{low2}(v))\in{B(v)}$.
\item $\mathit{high}(v):=\mathit{max}\{y\mid \exists (x,y)\in{B(v)}\}$.
\item $\mathit{highD}(v):=$ a vertex $x$ such that $(x,\mathit{high}(v))\in{B(v)}$.
\item $M(v):=\mathit{nca}\{x\mid\exists (x,y)\in{B(v)}\}$.
\item $\tilde{M}(v):=\mathit{nca}\{x\mid\exists (x,y)\in{B(v)} \mbox{ and } x \mbox{ is a proper descendant of } M(v)\}$.
\item $M_{low1}(v):=\mathit{nca}\{x\mid\exists (x,y)\in{B(v)} \mbox{ and } x \mbox{ is a descendant of the } \mathit{low1} \mbox{ child of } M(v)\}$.
\item $M_{low2}(v):=\mathit{nca}\{x\mid\exists (x,y)\in{B(v)} \mbox{ and } x \mbox{ is a descendant of the } \mathit{low2} \mbox{ child of } M(v)\}$.
\item $\mathit{nextM}(v):=$ the maximum vertex $v'<v$ such that $M(v')=M(v)$.
\item $\mathit{prevM}(v):=$ the minimum vertex $v'>v$ such that $M(v')=M(v)$.
\end{itemize}

We note that the notion of low points plays central role in classic algorithms for computing the biconnected components~\cite{dfs:t}, the triconnected components~\cite{3-connectivity:ht}
and the $3$-edge-connected components~\cite{GI:ECtoVC,3-connectivity:ht,NagamochiIbaraki:3CC,Tsin:3CC} of a graph.
Hopcroft and Tarjan~\cite{3-connectivity:ht} also use a concept of high points, which, however, is different from ours.
Our goal is to provide an method to compute type-$2$ cuts that avoids the use of $\mathit{high}(v)$ and $\mathit{highD}(v)$.
We achieve this by introducing two new parameters.

\begin{figure*}[t!]
\begin{center}
\centerline{\includegraphics[trim={0 0 0 6.2cm}, clip=true, width=\textwidth]{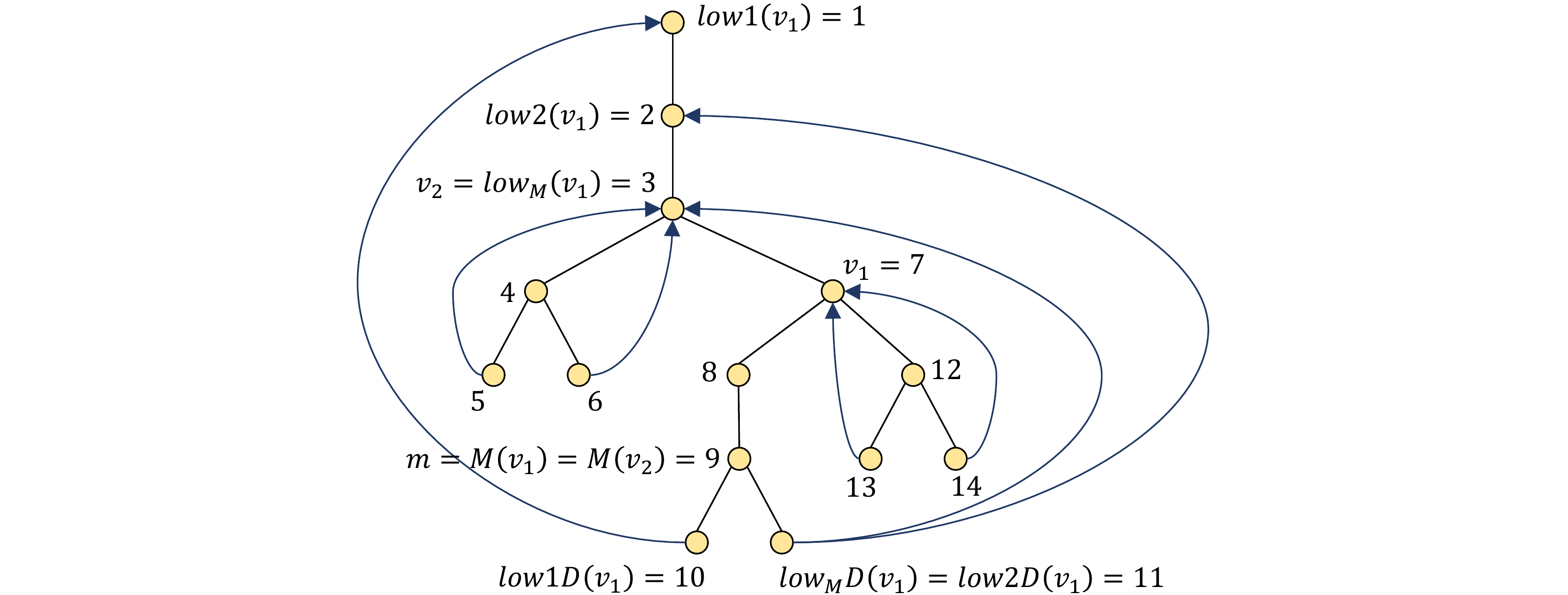}}
\caption{An illustration of the concepts defined on a depth-first search (DFS) spanning tree of an undirected graph.
Vertices are numbered in DFS order and back-edges are shown directed from descendant to ancestor in the DFS tree.
Vertices $v_1$ and $v_2$ have $M(v_1)=M(v_2)=m$, hence $\{v_1,v_2\} \subseteq M^{-1}(m)$, $\mathit{nextM}(v_1)=v_2$, and $\mathit{prevM}(v_2)=v_1$.
Moreover, $(11,3) \not\in B(v_2)$, so $\mathit{low_M}(v_1)=3$.
\label{figure:dfs-example}}
\end{center}
\vspace{-1cm}
\end{figure*}

\subsection{Two new key parameters}

Here we assume that the input graph $G$ is $3$-edge-connected. Let $v$ be a vertex such that $\mathit{nextM}(v)\neq\emptyset$. Then we have $B(\mathit{nextM}(v))\subset B(v)$. Thus we can define $\mathit{low_M}(v)$ as the lowest lower end of all back-edges in $B(v)\setminus{B(\mathit{nextM}(v))}$, and we let $\mathit{low_{M}D}(v)$ be a vertex such that $(\mathit{low_{M}D}(v),\mathit{low_M}(v))$ is a back-edge in $B(v)\setminus{B(\mathit{nextM}(v))}$. Formally, we have
\begin{itemize}
\item $\mathit{low_M}(v):=\mathit{min}\{y\mid\exists (x,y)\in B(v)\setminus B(\mathit{nextM}(v))\}$.
\item $\mathit{low_{M}D}(v):=$ a vertex $x$ such that $(x,\mathit{low_M}(v))\in B(v)$.
\end{itemize}

Now we describe how to compute $(\mathit{low_{M}D}(v),\mathit{low_M}(v))$ for every vertex $v$ such that $\mathit{nextM}(v)\neq\emptyset$.
To do this efficiently, we process the vertices in a bottom-up fashion. For every vertex $v$ that we process, we check whether $u=\mathit{prevM}(v)\neq\emptyset$. If that is the case, then $\mathit{low_M}(u)$ is defined and it lies on the simple tree-path $T(u,v]$. Thus we descend the path $T(u,v]$, starting from $v$, following the $\mathit{low1}$ children of the vertices on the path; for every vertex $y$ that we encounter we check whether there exists a back-edge $(x,y)$ with $x\in T(M(v))$. The first $y$ with this property is $\mathit{low_M}(u)$, and we set $(\mathit{low_{M}D}(u),\mathit{low_M}(u))\leftarrow(x,y)$.
To achieve linear running time, we let $\mathit{In}[y]$ denote the list of all vertices $x$ for which there exists an incoming back-edge to $y$ with higher end $x$. In other words, $\mathit{In}[y]$ contains all vertices $x$ for which there exists a back-edge $(x,y)$. Furthermore, we have the elements of $\mathit{In}[y]$ sorted in increasing order (this can be done easily in linear time with bucket-sort). When we process a vertex $y$ as we descend $T(u,v]$, during the processing of $v$, we traverse $\mathit{In}[y]$ starting from the element we accessed the last time we traversed $\mathit{In}[y]$ (or, if this is the first time we traverse $\mathit{In}[y]$, from the first element of $\mathit{In}[y]$). Thus, we need a variable $\mathit{currentBackEdge}[y]$ to store the element of $\mathit{In}[y]$ we accessed the last time we traversed $\mathit{In}[y]$. Now, for every $x\in\mathit{In}[y]$ that we meet, we check whether $x\in T(M(v))$. If that is the case, then we set $(\mathit{low_{M}D}(u),\mathit{low_M}(u))\leftarrow(x,y)$; otherwise, we move to the next element of $\mathit{In}[y]$. If we reach the end of $\mathit{In}[y]$, then we descend the path $T(u,v]$ by moving to the $\mathit{low1}$ child of $y$. In fact, if $\mathit{prevM}(c_1(y))\neq\emptyset$, then we may descend immediately to $\mathit{low_M}(\mathit{prevM}(c_1(y)))$. This ensures that $\mathit{In}[y]$ will not be accessed again. Algorithm \ref{algorithm:lowM} shows how to compute all pairs $(\mathit{low_{M}D}(v),\mathit{low_M}(v))$, for all vertices $v$ with $\mathit{nextM}(v)\neq\emptyset$, in total linear time.

\begin{algorithm}[h!]
\caption{\textsf{Calculate $(\mathit{low_{M}D}(v),\mathit{low_M}(v))$, for every vertex $v$ with $\mathit{nextM}(v)\neq\emptyset$}}
\label{algorithm:lowM}
\LinesNumbered
\DontPrintSemicolon
calculate $\mathit{In}[y]$, for every vertex $y$, and have its elements sorted in increasing order\;
\lForEach{vertex $y$}{$\mathit{currentBackEdge}[y]\leftarrow$ first element of $\mathit{In}[y]$}
\lForEach{vertex $v$}{$\mathit{low_M}[v]\leftarrow \emptyset$}
\For{$v\leftarrow n$ to $v\leftarrow 1$}{
  \lIf{$\mathit{prevM}(v)=\emptyset$}{\textbf{continue}}
  $u \leftarrow \mathit{prevM}(v)$\;
  $y \leftarrow v$\;
  \While{$\mathit{low_M}(u)=\emptyset$}{
    \While{$\mathit{currentBackEdge}[y]\neq\emptyset$}{
      \label{lowM:line_1}
      $x \leftarrow \mathit{currentBackEdge}[y]$\;
      \If{$x<M(u)$}{
        $\mathit{currentBackEdge}[y]\leftarrow$ next element of $\mathit{In}[y]$\;
      }
      \Else{
        \If{$x<M(u)+\mathit{ND}(M(u))$}{
          $(\mathit{low_{M}D}(u),\mathit{low_M}(u))\leftarrow(x,y)$\;
        }
        \textbf{break}\;
      }
    }
    \If{$\mathit{low_M}(u)=\emptyset$}{
      \lIf{$\mathit{prevM}(c_1(y))=\emptyset$}{$y\leftarrow c_1(y)$}
      \lElse{$y\leftarrow\mathit{low_M}(\mathit{prevM}(c_1(y)))$}
    }
  }
}
\end{algorithm}

\begin{proposition}
Algorithm \ref{algorithm:lowM} correctly computes all pairs $(\mathit{low_{M}D}(v),\mathit{low_M}(v))$, for all vertices $v$ with $\mathit{nextM}(v)\neq\emptyset$, in total linear time.
\end{proposition}
\begin{proof}
Let $v$ be a vertex with $\mathit{prevM}(v)=u\neq\emptyset$. We will prove inductively that $(\mathit{low_{M}D}(u),\mathit{low_M}(u))$ will be computed correctly, and that $\mathit{low_{M}D}(u)$ will be the lowest vertex which is a descendant of $M(u)$ such that $(\mathit{low_{M}D}(u),\mathit{low_M}(u))$ is a back-edge. So let us assume that we have run Algorithm \ref{algorithm:lowM} and we have correctly computed all pairs $(\mathit{low_{M}D}(u'),\mathit{low_M}(u'))$, for all vertices $v'>v$ with $\mathit{prevM}(v')=u'\neq\emptyset$, and $\mathit{low_{M}D}(u')$ is the lowest vertex in $T(M(u'))$ such that $(\mathit{low_{M}D}(u'),\mathit{low_M}(u'))$ is a back-edge. Suppose also that we have currently descended the path $T(u,v]$, we have reached $y$, and $\mathit{low_M}(v)\geq y$.

Let us assume, first, that $\mathit{low_M}(v)=y$, and let $(x,y)$ be the back-edge such that $x\in T(M(v))$ and $x$ is minimal with this property. The \textbf{while} loop in line \ref{lowM:line_1} will search the list of incoming back-edges to $y$, starting from $\mathit{currentBackEdge}[y]$. If $\mathit{currentBackEdge}[y]$ is the first element of $\mathit{In}[y]$, then is it certainly true that $x$ will be found. Otherwise, let $x'=\mathit{currentBackEdge}[y]$. Due to the inductive hypothesis, we have that $(x',y)=(\mathit{low_{M}D}(u'),\mathit{low_M}(u'))$, for a vertex $u'$ with $\mathit{nextM}(u')=v'>v$. Then, $y$ is in $T(u',v']$, but also in $T(u,v]$, and thus it is a common descendant of $v$ and $v'$. This means that $v$ and $v'$ are related as ancestor and descendant. In particular, since $v'>v$, we have that $v$ is an ancestor of $v'$. Furthermore, since $y$ is an ancestor of $u$, it is also an ancestor of $M(u)=M(v)$; therefore, since $v'$ is an ancestor of $y$, it is also an ancestor of $M(v)$. Since $v$ is an ancestor of $v'$, this implies that $M(v')$ is an ancestor of $M(v)$. Since $M(v')=M(u')$ and $M(v)=M(u)$, we thus have that $M(u')$ is an ancestor of $M(u)$, and therefore $M(u')\leq M(u)$. Thus, since $x'$ is the lowest descendant of $M(u')$ such that $(x',y)$ is a back-edge, and $x$ is the lowest descendant of $M(u)$ such that $(x,y)$ is a back-edge, we have $x'\leq x$. This shows that $x$ will be accessed during the \textbf{while} loop in line \ref{lowM:line_1}.

Now let us assume that $\mathit{low_M}(v)\neq y$. This means that $\mathit{low_M}(u)$ is greater than $y$, and we have to descend the path $T(u,y)$ to find it. First, let $c$ be the child of $y$ in the direction of $u$. Then we have $\mathit{low1}(c)<v$ (since $M(v)=M(u)$ is a descendant of $u$, and therefore a descendant of $c$, and we have $\mathit{low1}(M(v))<v$). If there was another child $c'$ of $y$ with $\mathit{low1}(c')<v$, this would imply that $M(v)=y$, which is absurd, since $y$ is a proper ancestor of $u$, and therefore a proper ancestor of $M(u)=M(v)$. This means that $c$ is the $\mathit{low1}$ child of $v$, and thus we may descend to $c_1(y)=y'$. Now we have $\mathit{low_M}(u)\geq y'$. If $\mathit{prevM}(y')=\emptyset$, then we simply traverse the list of incoming back-edges to $y'$, in line \ref{lowM:line_1}, and repeat the same process. Otherwise, let $u'=\mathit{prevM}(y')$. Due to the inductive hypothesis, we know that $\mathit{low_M}(u')$ has been computed correctly. Since $y'$ is an ancestor of $u$, it is also an ancestor of $M(u)=M(v)$. Furthermore, $y'$ is a descendant of $v$. Thus, $M(y')$ is an ancestor of $M(v)$, and therefore $M(u')$ is an ancestor of $M(u)$ (since $M(y')=M(u')$ and $M(v)=M(u)$). This means that $u'$ is an ancestor of $M(u)$. Now we see that $\mathit{low_M}(u)$ lies on $T(u,\mathit{low_M}(u')]$. (For otherwise, $(\mathit{low_{M}D}(u),\mathit{low_M}(u))$ would be a back-edge in $B(u')$ with $\mathit{low_M}(u)\geq y'=\mathit{nextM}(u')$ and $\mathit{low_M}(u)<\mathit{low_M}(u')$, contradicting the minimality of $\mathit{low_M}(u')$). Thus we may descend immediately to $\mathit{low_M}(u')$. Then we traverse the list of incoming back-edges to $\mathit{low_M}(u')$, in line \ref{lowM:line_1}, and repeat the same process. Eventually we will reach $\mathit{low_M}(u)$ and have it computed correctly. It should be clear that no vertex on the path $T(\mathit{low_M}(u),v)$ will be traversed again, and this ensures the linear complexity of Algorithm~\ref{algorithm:lowM}.
\end{proof}

\section{Simple algorithm for computing all $3$-cuts of type-$2$}
\label{sec:simple-type2}

In this section we will show how to compute all $3$-cuts of type-$2$ (consisting of two tree-edges and one back-edge) of a $3$-edge-connected graph in linear time, without using the $\mathit{high}$ points of  \cite{Linear4ECC:ESA21}.
We use the following characterization of such cuts.

\begin{lemma}(\cite{Linear4ECC:ESA21})
\label{lemma:type-2-cut}
Let $u,v$ be two vertices with $v<u$. Suppose that $\{(u,p(u)),(v,p(v)),e\}$ is a $3$-cut, where $e$ is a back-edge.
Then $v$ is an ancestor of $u$, and either $(1)$ $B(v)=B(u)\sqcup \{e\}$ or $(2)$ $B(u)=B(v)\sqcup \{e\}$. Conversely, if there exists a back-edge $e$ such that $(1)$ or $(2)$ is true, then $\{(u,p(u)),(v,p(v)),e\}$ is a $3$-cut.
\end{lemma}

In the following, for any vertex $u$, $V(u)$ denotes the set of all vertices $v$ that are ancestors of $u$ and such that $B(v)=B(u)\sqcup\{e\}$, for a back-edge $e$. Similarly, for any vertex $v$, $U(v)$ denotes the set of all vertices $u$ that are descendants of $v$ and such that $B(u)=B(v)\sqcup\{e\}$, for a back-edge $e$.
In \cite{Linear4ECC:ESA21} it is shown that  $V(u)\cap V(u')=\emptyset$ (resp. $U(v)\cap U(v')=\emptyset$) for every two vertices $u\neq u'$ (resp. $v\neq v'$).
Thus, in order to find all type-$2$ cuts, it is sufficient to find, for every vertex $u$ (resp. $v$) the unique vertex $v$ (resp. $u$), if it exists, such that $u\in U(v)$ (resp. $v\in V(u)$), and then identify the back-edge $e$ such that $\{(u,p(u)),(v,p(v)),e\}$ is a $3$-cut. The following two lemmas show how to identify $e$. 

\begin{lemma}
\label{lemma:identify_e_1}
Let $u,v$ be two vertices such that $u$ is a descendant of $v$ and $B(v)=B(u)\sqcup\{e\}$, for a back-edge $e=(x,y)$. Then we have $M(u)\in\{\tilde{M}(v),M_{low1}(v),M_{low2}(v)\}$. In particular, we have that either $(1)$ $M(u)=\tilde{M}(v)$ and $e=(M(v),l_1(M(v)))$, or $(2)$ $M(u)=M_{low1}(v)$ and $e=(M_{low2}(v),l_1(M_{low2}(v)))$, or $(3)$ $M(u)=M_{low2}(v)$ and $e=(M_{low1}(v),l_1(M_{low1}(v)))$.
\end{lemma}
\begin{proof}
First we will show that $M(v)$ is a proper ancestor of $M(u)$. Obviously, $M(v)$ is an ancestor of $M(u)$, since $B(v)=B(u)\sqcup\{e\}$. Furthermore, since $e\in B(v)$, $x$ is a descendant of $M(v)$, and $y$ is a proper ancestor of $v$, and therefore a proper ancestor of $u$. Thus, it cannot be the case that $M(u)=M(v)$, for otherwise we would have $e\in B(u)$. This shows that $M(v)$ is a proper ancestor of $M(u)$.
Now we will show that $\mathit{low1}(c_k(M(v)))\geq v$, for every $k>2$. Suppose for the sake of contradiction, that there exists a $k>2$ such that $\mathit{low1}(c_k(M(v)))<v$. Then we have $\mathit{low1}(c_k'(M(v)))<v$, for every $k'\leq k$, and so $B(v)\cap B(c_k'(M(v)))\neq\emptyset$, for every $k'\leq k$. Then, since $B(u)=B(v)\setminus\{e\}$, we have that $B(u)\cap B(c_k'(M(v)))\neq\emptyset$, for all $k'\leq k$, except possibly a $\tilde{k}\in\{1,\dotsc,k\}$. Thus, $M(u)$ is a common ancestor of all $\{c_1(M(v)),\dotsc,c_k(M(v))\}$, except possibly $c_{\tilde{k}}(M(v))$, and so, since $k>2$, we conclude that $M(u)$ is an ancestor of $M(v)$, which is absurd. Thus we have demonstrated that $\mathit{low1}(c_k(M(v)))\geq v$, for every $k>2$.

Now, there are two cases to consider: either $x=M(v)$, or $x$ is a descendant of a child of $M(v)$. First take the case $x=M(v)$. Then $(M(v),l_1(M(v)))$ is obviously a back-edge in $B(v)$. Furthermore, since $M(u)$ is not an ancestor of $M(v)$, we also have $(M(v),l_1(M(v)))\notin B(u)$. Thus $e=(M(v),l_1(M(v)))$. Since every other back-edge of the form $(M(v),y')$ with $y'<v$ must have $(M(v),y')\in B(v)\setminus B(u)$, we conclude that $e$ is the unique back-edge of the form $(M(v),y')$ with $y'<v$. Since $B(u)=B(v)\setminus\{e\}$, this means that $M(u)=\mathit{nca}(\{x'\mid \exists (x',y')\in B(v)\}\setminus\{M(v)\})=\tilde{M}(v)$.
Now consider the case that $x$ is a descendant of a child of $M(v)$. Then we have that $x$ is either a descendant of $c_1(M(v))$, or a descendant of $c_2(M(v))$ (since $\mathit{low1}(c_k(M(v)))\geq v$, for every $k>2$). We will consider only the case that $x$ is a descendant of $c_1(M(v))$, since the other case can be treated in a similar manner. So let $x$ be a descendant of $c_1(M(v))$. Then we must have $\mathit{low1}(c_2(M(v)))<v$, for otherwise there would exist a back-edge $e'$ of the form $(M(v),y')$ with $y'<v$, and so we would have two distinct back-edges $e,e'\in B(v)\setminus B(u)$, which is absurd. Thus, $B(v)\cap B(c_2(M(v)))\neq\emptyset$, and therefore, since $B(u)=B(v)\setminus\{e\}$, we have $B(u)\cap B(c_2(M(v)))\neq\emptyset$. Thus, we must necessarily have $B(u)\cap B(c_1(M(v)))=\emptyset$, for otherwise $M(u)$ would be an ancestor of both $c_1(M(v))$ and $c_2(M(v))$, and therefore an ancestor of $M(v)$, which is absurd. Since $B(v)=B(u)\sqcup\{e\}$ and $\mathit{low1}(c_1(M(v)))<v$, this means that $B(v)\cap B(c_1(M(v)))=\{e\}$, and therefore $e=(M_{low1}(v),l_1(M_{low1}(v)))$.
\end{proof}

\begin{lemma}
\label{lemma:identify_e_2}
Let $u,v$ be two vertices such that $u$ is a descendant of $v\neq r$ and $B(u)=B(v)\sqcup\{e\}$, for a back-edge $e=(x,y)$. Then we have $M(v)\in\{M(u),\tilde{M}(u),M_{low1}(u)\}$. In particular, we have that either $(1)$ $M(v)=M(u)$ and $e=(\mathit{low_M{D}}(u),\mathit{low_M}(u))$, or $(2)$ $M(v)=\tilde{M}(u)$ and $e=(M(u),l_1(M(u)))$, or $(3)$ $M(v)=M_{low1}(u)$ and $e=(M_{low2}(u),l_1(M_{low2}(u)))$.
\end{lemma}
\begin{proof}
$B(u)=B(v)\sqcup\{e\}$ implies that $M(u)$ is an ancestor of $M(v)$. If $M(u)=M(v)$, then $v=\mathit{nextM}(u)$. (For otherwise, there exists a vertex $v'$ with $M(v')=M(u)$ and $v<v'<u$, and so we have $B(v)\subset B(v')\subset B(u)$ - which is impossible, since $B(u)=B(v)\sqcup\{e\}$.) Since $e\in B(u)\setminus B(v)$ and $|B(u)\setminus B(v)|=1$ and $(\mathit{low_M{D}}(u),\mathit{low_M}(u))$ is a back-edge in $B(u)\setminus B(\mathit{nextM}(u))$, we conclude that $e=(\mathit{low_M{D}}(u),\mathit{low_M}(u))$.

Now let's assume that $M(u)$ is a proper ancestor of $M(v)$. There are two cases to consider: either $x=M(u)$, or $x$ is a descendant of a child of $M(u)$. If $x=M(u)$, then $(M(u),l_1(M(u)))$ is a back-edge in $B(u)$. Furthermore, $(M(u),l_1(M(u)))\notin B(v)$ (for otherwise we would have that $M(v)$ is an ancestor of $M(u)$). This shows that $e=(M(u),l_1(M(u)))$. We also see that $(M(u),y')\notin B(v)$, for any back-edge $(M(u),y')$ with $y'<u$. Thus, since $B(v)=B(u)\setminus\{e\}$, we have $M(v)=\mathit{nca}(\{x'\mid \exists (x',y')\in B(u)\}\setminus\{M(u)\})=\tilde{M}(u)$. Finally, let's assume that $x$ is a descendant of a child of $M(u)$, i.e. $x$ is a descendant of $c_k(M(u))$, for some $k$. We will show that $\mathit{low1}(c_k'(M(u)))<v$, for every $k'<k$. So suppose for the sake of contradiction, that $\mathit{low1}(c_k'(M(u)))\geq v$, for some $k'<k$. Since $e\in B(c_k(M(u)))$, we have $\mathit{low1}(c_k(M(u)))<u$; therefore, since $\mathit{low1}(c_k'(M(u)))\leq\mathit{low1}(c_k(M(u)))$, we have $\mathit{low1}(c_k'(M(u)))<u$. This shows that there exists a back-edge $e'\in B(u)$ which is also in $B(c_k'(M(u)))$. But since, $\mathit{low1}(c_k'(M(u)))\geq v$, it cannot be the case that $e'\in B(v)$. Thus we have provided two distinct back-edges $e,e'\in B(u)\setminus B(v)$, which is absurd. This shows that $\mathit{low1}(c_k'(M(u)))<v$, for every $k'<k$. If $k>2$, this implies that $M(v)$ is a common ancestor of at least two children of $M(u)$, which is absurd. Thus we have that $k\leq 2$. Now suppose for the sake of contradiction, that $k=1$. It cannot be the case that $\mathit{low1}(c_2(M(u)))<v$, for otherwise $\mathit{low1}(c_1(M(u)))<v$ also, which would imply that $M(v)$ is an ancestor of $M(u)$, which is absurd. Now, it cannot be the case that $\mathit{low1}(c_2(M(u)))<u$, for otherwise there would exist two distinct back-edges $e',e\in B(u)\setminus B(v)$, which is also absurd. Thus, $c_1(M(u))$ is the only child of $M(u)$ with $\mathit{low1}(c_1(M(u)))<u$, which means that there must exist a back-edge $(M(u),y')$ with $y'<u$. Now if $y'<v$, $M(v)$ is an ancestor of $M(u)$, which is absurd. And if $y\geq v$, then we have two distinct back-edges $e,e'\in B(u)\setminus B(v)$, which is also absurd. Thus the assumption $k=1$ cannot hold, and we conclude that $k=2$, i.e. $x$ is a descendant of $c_2(M(u))$. Now it cannot be the case that $\mathit{low1}(c_2(M(u)))<v$, for this would imply $\mathit{low1}(c_1(M(u)))<v$, and so we would have that $M(v)$ is an ancestor of $M(u)$. Thus $v<\mathit{low1}(c_2(M(u)))<u$, and so $B(u)\cap B(c_2(M(u)))=\{e\}$ (since $|B(u)\setminus B(v)|=1$). This shows that $e=(M_{low2}(u),l_1(M_{low2}(u)))$. Then, since $B(v)=B(u)\setminus\{e\}$, we have that $B(v)=(B(u)\cap B(c_1(M(u))))\sqcup(B(u)\cap\{(M(u),y')\in B(u)\})$. But since $M(u)$ is a proper ancestor of $M(v)$, we conclude that $B(v)=B(u)\cap B(c_1(M(u)))$, and therefore $M(v)=M_{low1}(u)$.
\end{proof}

The next two lemmas are the basis for a linear-time algorithm to compute all $3$-cuts of type-$2$.

\begin{lemma}(\cite{Linear4ECC:ESA21})
\label{lemma:find_u_from_v}
Let $v$ be an ancestor of $u$ such that $m=M(u) \in \{\tilde{M}(v), M_{low1}(v), M_{low2}(v)\}$.
Then, $v\in V(u)$ if and only if $u$ is the minimum vertex in $M^{-1}(m)$ greater than $v$ and such that $\mathit{high}(u)<v$ and $\mathit{b\_count}(v)=\mathit{b\_count}(u)+1$.
\end{lemma}

\begin{lemma}(\cite{Linear4ECC:ESA21})
\label{lemma:find_v_from_u}
Let $u$ be a descendant of $v$ such that $m=M(v) \in \{M(u), \tilde{M}(u), M_{low1}(u)\}$.
Then $u\in U(v)$ if and only if $v$ is the maximum vertex in $M^{-1}(m)$ less than $u$ and such that $\mathit{b\_count}(u)=\mathit{b\_count}(v)+1$.
\end{lemma}

Now let $\{(u,p(u)),(v,p(v)),e\}$ be a $3$-cut (of type-$2$), where $e$ is a back-edge and $v<u$.
By Lemma~\ref{lemma:type-2-cut},
we have that $u$ is a descendant of $v$ and that either $B(v)=B(u)\sqcup\{e\}$ or $B(u)=B(v)\sqcup\{e\}$. We will handle these cases in turn.
In the following, we let $e=(x,y)$ be the back-edge of the $3$-cut of type-$2$.

\paragraph*{\textcolor{lgray}{$\blacktriangleright$} Case $B(v)=B(u)\sqcup\{e\}$}
By Lemma \ref{lemma:identify_e_1}, one of the following cases holds (see Figure~\ref{figure:type2-3cuts-example}):
either $(1)$ $x=M(v)$, $y=l_1(M(v))$ and $M(u)=\tilde{M}(v)$, or $(2)$ $x=M_{low2}(v)$, $y=l_1(M_{low2}(v))$ and $M(u)=M_{low1}(v)$, or $(3)$ $x=M_{low1}(v)$, $y=l_1(M_{low1}(v))$ and $M(u)=M_{low2}(v)$. In any case, by Lemma~\ref{lemma:find_u_from_v} we have that $u$ is the minimum vertex in $M^{-1}(m)$ greater than $v$, where $m=\tilde{M}(v)$, or $m=M_{low1}(v)$, or $m=M_{low2}(v)$, depending on whether $(1)$, or $(2)$, or $(3)$ is true, respectively. Thus, we can compute all those $3$-cuts by finding, for every vertex $v$, and every $m\in\{\tilde{M}(v),M_{low1}(v),M_{low2}(v)\}$, the minimum vertex $u$ in $M^{-1}(m)$ greater than $v$, and then check whether $B(v)=B(u)\sqcup\{e\}$, for a back-edge $e$. This last condition is equivalent to $\mathit{b\_count}(v)=\mathit{b\_count}(u)+1$ and $\mathit{high}(u)<\mathit{high}(v)$. Note that $\mathit{high}(u)<\mathit{high}(v)$ is necessary to ensure $B(u)\subseteq B(v)$; then, $\mathit{b\_count}(v)=\mathit{b\_count}(u)+1$ implies the existence of a back-edge $e$ such that $B(v)=B(u)\sqcup\{e\}$. We will now show how to check this condition without using the $\mathit{high}$ points.
To that end, we provide the following characterizations.

\begin{figure*}[t!]
\begin{center}
\centerline{\includegraphics[trim={0 0 0 4.7cm}, clip=true, width=\textwidth]{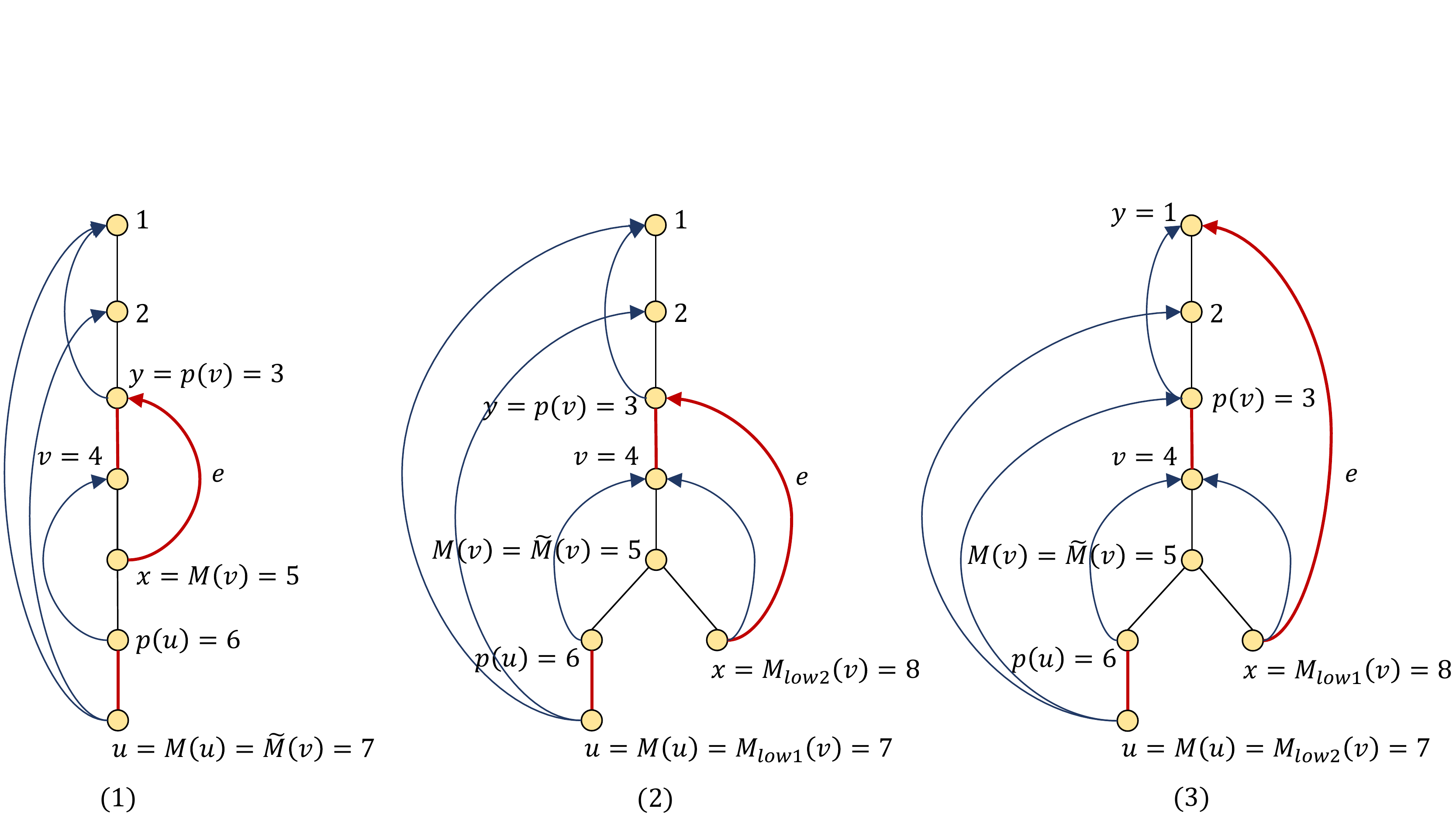}}
\caption{An illustration of the three cases for type-$2$ cuts $\{(u,p(u)),(v,p(v)),e\}$ (shown with red edges) where $v$ is an ancestor of $u$, $e$ is a back-edge, and $B(v)=B(u)\sqcup\{e\}$.
Vertices are numbered in DFS order and back-edges are shown directed from descendant to ancestor in the DFS tree.
\label{figure:type2-3cuts-example}}
\end{center}
\vspace{-1cm}
\end{figure*}

\begin{lemma} \textcolor{cgray}{(Case (1))}
\label{lemma:find_u_from_v_Mtilde}
Let $m=\tilde{M}(v)$ and $u$ be the minimum vertex in $M^{-1}(m)$ strictly greater than $v$. Then there exists a back-edge $e$ such that $B(v)=B(u)\sqcup\{e\}$ if and only if: $\mathit{b\_count}(v)=\mathit{b\_count}(u)+1$, $l_2(M(v))\geq v$, and either $M(v)$ has no $\mathit{low2}$ child, or $\mathit{low1}(c_2(M(v)))\geq v$.
\end{lemma}
\begin{proof}
($\Rightarrow$) $\mathit{b\_count}(v)=\mathit{b\_count}(u)+1$ is an immediate consequence of $B(v)=B(u)\sqcup\{e\}$. Furthermore, $B(v)=B(u)\sqcup\{e\}$ implies that $M(v)$ is an ancestor of $M(u)$. But it cannot be the case that $M(v)=M(u)$ (for otherwise, $v$ being an ancestor of $u$ would imply that $B(v)$ is a subset of $B(u)$); thus, $M(v)$ is a proper ancestor of $M(u)$. Since $M(u)=\tilde{M}(v)$, this implies that there is no back-edge $(x,y)$ with $x$ a descendant of a child $c$ of $M(v)$, with $c\neq c_1(M(v))$, and $y$ a proper ancestor of $v$ (otherwise, we would have $\tilde{M}(v)=M(v)$). This means that $\mathit{low1}(c)\geq v$, for every child $c$ of $M(v)$ with $c\neq c_1(M(v))$. In other words, either $M(v)$ has no $\mathit{low2}$ child, or $\mathit{low1}(c_2(M(v)))\geq v$. Since $\tilde{M}(v)\neq\emptyset$, this also means that there exists a back-edge $\tilde{e}=(M(v),l_1(M(v)))$. Obviously, $\tilde{e}=e$, since $\tilde{e}\in B(v)\setminus{B(u)}$. Now, if we had $l_2(M(v))<v$, then there would exist a back-edge $e'=(M(v),l_2(M(v)))\neq e$. But then we would have $e'\in B(v)\setminus{B(u)}$, contradicting $B(v)=B(u)\sqcup\{e\}$.\\
($\Leftarrow$) Let $(x,y)\in B(v)$. Since $M(v)$ either has no $\mathit{low2}$ child, or $\mathit{low1}(c_2(M(v)))\geq v$, we have that $x$ is either $M(v)$ or a descendant of $\tilde{M}(v)$.
Let $\tilde{B}(v)=\{(x,y)|x\in T(\tilde{M}(v)) \mbox{ and } y<v\}$. Then we have $B(v)=\tilde{B}(v)\sqcup\{(M(v),z)|z<v\}$. Now, if $(x,y)\in\tilde{B}(v)$, then $x$ is a descendant of $\tilde{M}(v)$, and therefore a descendant of $M(u)$. Furthermore, $y$ is an ancestor of $v$, and therefore an ancestor of $u$. This means that $\tilde{B}(v)\subseteq B(u)$. Now, since $l_2(M(v))\geq v$, there is at most one back-edge $e=(M(v),z)$ with $z<v$ (which thus satisfies $z=l_1(M(v))$). But such a back-edge must necessarily exist, for otherwise we would have $B(v)=\tilde{B}(v)\subseteq B(u)$, contradicting $\mathit{b\_count}(v)=\mathit{b\_count}(u)+1$. Now, since $B(v)=\tilde{B}(v)\sqcup\{(M(v),z)|z<v\}$ and $|\{(M(v),z)|z<v\}|=1$ and $\tilde{B}(v)\subseteq B(u)$ and $\mathit{b\_count}(v)=\mathit{b\_count}(u)+1$, we conclude that $\tilde{B}(v)=B(u)$, and therefore $B(v)=B(u)\sqcup\{(M(v),l_1(M(v))\}$.
\end{proof}

\begin{lemma} \textcolor{cgray}{(Case (2))}
\label{lemma:find_u_from_v_Mlow1}
Let $l_1(M(v))\geq v$, $m=M_{low1}(v)$, and $u$ be the minimum vertex in $M^{-1}(m)$ strictly greater than $v$. Then there exists a back-edge $e$ such that $B(v)=B(u)\sqcup\{e\}$ if and only if: $\mathit{b\_count}(v)=\mathit{b\_count}(u)+1$, $\mathit{low2}(M_{low2}(v))\geq v$, and either $M(v)$ has no $\mathit{low3}$ child, or $\mathit{low1}(c_3(M(v)))\geq v$.
\end{lemma}
\begin{proof}
($\Rightarrow$) $\mathit{b\_count}(v)=\mathit{b\_count}(u)+1$ is an immediate consequence of $B(v)=B(u)\sqcup\{e\}$. Now, $l_1(M(v))\geq v$ implies that every back-edge $(x,y)\in B(v)$ has $x\in T(c)$, where $c$ is a child of $M(v)$. Take a back-edge $(x,y)\in B(v)$, with $x\notin T(c_1(M(v)))$. Then $(x,y)\notin B(u)$, and therefore $B(v)=B(u)\sqcup\{e\}$ implies that $(x,y)=e$. This means that $x\in T(c_2(M(v))$, $\mathit{low2}(M_{low2}(v))\geq v$, and either $M(v)$ has no $\mathit{low3}$ child, or $\mathit{low1}(c_3(M(v)))\geq v$.\\
($\Leftarrow$) $l_1(M(v))\geq v$ implies that every back-edge $(x,y)\in B(v)$ has $x\in T(c)$, where $c$ is a child of $M(v)$. Furthermore, it implies that there exists at least one back-edge $(x,y)\in B(v)$ with $x\in T(c_2(M(v)))$. Now let $B_{low1}(v)=\{(x,y)|x\in T(c_1(M(v)))\}$. Then, since $\mathit{low2}(M_{low2}(v))\geq v$, and either $M(v)$ has no $\mathit{low3}$ child, or $\mathit{low1}(c_3(M(v)))\geq v$, we have that $B(v)=B_{low1}(v)\sqcup\{(M_{low2}(v),l_1(M_{low2}(v))\}$. Since $M_{low1}(v)=M(u)$ and $v$ is an ancestor of $u$, we have that $B_{low1}(v)\subseteq B(u)$. Now, from $B(v)=B_{low1}(v)\sqcup\{(M_{low2}(v),l_1(M_{low2}(v))\}$, $B_{low1}(v)\subseteq B(u)$, and $\mathit{b\_count}(v)=\mathit{b\_count}(u)+1$, we conclude that $B(v)=B(u)\sqcup\{(M_{low2}(v),l_1(M_{low2}(v))\}$.
\end{proof}

\begin{lemma} \textcolor{cgray}{(Case (3))}
\label{lemma:find_u_from_v_Mlow2}
Let $l_1(M(v))\geq v$, $m=M_{low2}(v)$, and $u$ be the mimimum vertex in $M^{-1}(m)$ strictly greater than $v$. Then there exists a back-edge $e$ such that $B(v)=B(u)\sqcup\{e\}$ if and only if: $\mathit{b\_count}(v)=\mathit{b\_count}(u)+1$, $\mathit{low2}(M_{low1}(v))\geq v$, and either $M(v)$ has no $\mathit{low3}$ child, or $\mathit{low1}(c_3)\geq v$, where $c_3$ is the $\mathit{low3}$ child of $M(v)$.
\end{lemma}
\begin{proof}
($\Rightarrow$) $\mathit{b\_count}(v)=\mathit{b\_count}(u)+1$ is an immediate consequence of $B(v)=B(u)\sqcup\{e\}$. Now, $l_1(M(v))\geq v$ implies that every back-edge $(x,y)\in B(v)$ has $x\in T(c)$, where $c$ is a child of $M(v)$. Take a back-edge $(x,y)\in B(v)$, with $x\notin T(c_2(M(v)))$. Then $(x,y)\notin B(u)$, and therefore $B(v)=B(u)\sqcup\{e\}$ implies that $(x,y)=e$. This means that $x\in T(c_1(M(v))$, $\mathit{low2}(M_{low1}(v))\geq v$, and either $M(v)$ has no $\mathit{low3}$ child, or $\mathit{low1}(c_3(M(v)))\geq v$.\\
($\Leftarrow$) $l_1(M(v))\geq v$ implies that every back-edge $(x,y)\in B(v)$ has $x\in T(c)$, where $c$ is a child of $M(v)$. Furthermore, it implies that there exists at least one back-edge $(x,y)\in B(v)$ with $x\in T(c_1(M(v)))$. Now let $B_{low2}(v)=\{(x,y)|x\in T(c_2(M(v)))\}$. Then, since $\mathit{low2}(M_{low1}(v))\geq v$, and either $M(v)$ has no $\mathit{low3}$ child, or $\mathit{low1}(c_3(M(v)))\geq v$, we have that $B(v)=B_{low2}(v)\sqcup\{(M_{low1}(v),l_1(M_{low1}(v))\}$. Since $M_{low2}(v)=M(u)$ and $v$ is an ancestor of $u$, we have that $B_{low2}(v)\subseteq B(u)$. Now, from $B(v)=B_{low2}(v)\sqcup\{(M_{low1}(v),l_1(M_{low1}(v))\}$, $B_{low2}(v)\subseteq B(u)$, and $\mathit{b\_count}(v)=\mathit{b\_count}(u)+1$, we conclude that $B(v)=B(u)\sqcup\{(M_{low1}(v),l_1(M_{low1}(v))\}$.
\end{proof}

Algorithm~\ref{algorithm:B(v)=B(u)e_nohigh} shows how we can determine all $3$-cuts of this type.
The idea is to handle cases $(1)$, $(2)$ and $(3)$ separately, and our goal is to find, for every vertex $v$, the minimum vertex $u$ in $M^{-1}(m)$ (where $m\in\{\tilde{M}(v),M_{low1}(v),M_{low2}(v)\}$) which is strictly greater than $v$, and then check whether $B(v)=B(u)\sqcup\{e\}$, for a back-edge $e$. We can perform this search in linear time by processing the vertices in a bottom-up fashion, and keep in a variable $\mathit{currentVertex}[m]$ the lowest element of $M^{-1}(m)$ that we accessed so far.
Then we can easily check in constant time whether a pair of vertices $u,v$ such that $u$ is the minimum vertex in $M^{-1}(m)$ which is strictly greater than $v$ satisfies $B(v)=B(u)\sqcup\{e\}$, for a back-edge $e$, and also identify this back-edge.

\begin{algorithm}[h!]
\caption{\textsf{Find all $3$-cuts $\{(u,p(u)),(v,p(v)),e)\}$, where $u$ is a descendant of $v$ and $B(v)=B(u)\sqcup\{e\}$, for a back-edge $e$.}}
\label{algorithm:B(v)=B(u)e_nohigh}
\LinesNumbered
\DontPrintSemicolon
initialize an array $\mathit{currentVertex}$ with $n$ entries\;
\textcolor{cblue}{\tcp{\textit{Case (1): $m=\tilde{M}(v)$}}}
\lForEach{vertex $x$}{$\mathit{currentVertex}[x] \leftarrow x$}
\For{$v\leftarrow n$ to $v=1$}{
  $m \leftarrow \tilde{M}(v)$\;
  \lIf{$m=\emptyset$}{\textbf{continue}}
  \textcolor{cblue}{\tcp{\textit{find the minimum $u\in M^{-1}(m)$ that is greater than $v$}}}
  $u \leftarrow \mathit{currentVertex}[m]$\;
  \lWhile{$\mathit{nextM}(u)\neq\emptyset$ \textbf{and} $\mathit{nextM}(u)> v$}{$u \leftarrow \mathit{nextM}(u)$}
  $\mathit{currentVertex}[m] \leftarrow u$\;
  \textcolor{cblue}{\tcp{\textit{check the condition in Lemma~\ref{lemma:find_u_from_v_Mtilde}}}}
  \If{$\mathit{b\_count}(v)=\mathit{b\_count}(u)+1$ \textbf{and} $l_2(M(v))\geq v$ \textbf{and} ($c_2(M(v))=\emptyset$ \textbf{or} $\mathit{low1}(c_2(M(v)))\geq v$)}{
    mark the triplet $\{(u,p(u)),(v,p(v)),(M(v),l_1(M(v)))\}$
  }
}
\textcolor{cblue}{\tcp{\textit{Case (2): $m=M_{low1}(v)$}}}
\lForEach{vertex $x$}{$\mathit{currentVertex}[x] \leftarrow x$}
\For{$v\leftarrow n$ to $v=1$}{
  $m \leftarrow M_{low1}(v)$\;
  \lIf{$m=\emptyset$ \textbf{or} $l_1(M(v))<v$}{\textbf{continue}}
  \textcolor{cblue}{\tcp{\textit{find the minimum $u\in M^{-1}(m)$ that is greater than $v$}}}
  $u \leftarrow \mathit{currentVertex}[m]$\;
  \lWhile{$\mathit{nextM}(u)\neq\emptyset$ \textbf{and} $\mathit{nextM}(u)> v$}{$u \leftarrow \mathit{nextM}(u)$}
  $\mathit{currentVertex}[m] \leftarrow u$\;
  \textcolor{cblue}{\tcp{\textit{check the condition in Lemma~\ref{lemma:find_u_from_v_Mlow1}}}}
  \If{$\mathit{b\_count}(v)=\mathit{b\_count}(u)+1$ \textbf{and} $\mathit{low2}(M_{low2}(v))\geq v$ \textbf{and} ($c_3(M(v))=\emptyset$ \textbf{or} $\mathit{low1}(c_3(M(v)))\geq v$)}{
    mark the triplet $\{(u,p(u)),(v,p(v)),(M_{low2}(v),l_1(M_{low2}(v)))\}$
  }
}
\textcolor{cblue}{\tcp{\textit{Case (3): $m=M_{low2}(v)$}}}
\lForEach{vertex $x$}{$\mathit{currentVertex}[x] \leftarrow x$}
\For{$v\leftarrow n$ to $v=1$}{
  $m \leftarrow M_{low2}(v)$\;
  \lIf{$m=\emptyset$ \textbf{or} $l_1(M(v))<v$}{\textbf{continue}}
  \textcolor{cblue}{\tcp{\textit{find the minimum $u\in M^{-1}(m)$ that is greater than $v$}}}
  $u \leftarrow \mathit{currentVertex}[m]$\;
  \lWhile{$\mathit{nextM}(u)\neq\emptyset$ \textbf{and} $\mathit{nextM}(u)> v$}{$u \leftarrow \mathit{nextM}(u)$}
  $\mathit{currentVertex}[m] \leftarrow u$\;
  \textcolor{cblue}{\tcp{\textit{check the condition in Lemma~\ref{lemma:find_u_from_v_Mlow2}}}}
  \If{$\mathit{b\_count}(v)=\mathit{b\_count}(u)+1$ \textbf{and} $\mathit{low2}(M_{low1}(v))\geq v$ \textbf{and} ($c_3(M(v))=\emptyset$ \textbf{or} $\mathit{low1}(c_3(M(v)))\geq v$)}{
    mark the triplet $\{(u,p(u)),(v,p(v)),(M_{low1}(v),l_1(M_{low1}(v)))\}$
  }
}
\end{algorithm}

\clearpage

\paragraph*{\textcolor{lgray}{$\blacktriangleright$} Case $B(u)=B(v)\sqcup\{e\}$}
Now let $u$, $v$ be two vertices such that $v$ is an ancestor of $u$ with $B(u)=B(v)\sqcup\{e\}$, for a back-edge $e$.
By Lemma \ref{lemma:identify_e_2}, one of the following cases holds (see Figure~\ref{figure:type2-3cuts-example2}): either $(1)$ $M(v)=M(u)$ and $e=(\mathit{low_M{D}}(u),\mathit{low_M}(u))$, or $(2)$ $M(v)=\tilde{M}(u)$ and $e=(M(u),l_1(M(u)))$, or $(3)$ $M(v)=M_{low1}(u)$ and $e=(M_{low2}(u),l_1(M_{low2}(u))$. In any case, by Lemma~\ref{lemma:find_v_from_u} we have that $v$ is the maximum vertex in $M^{-1}(m)$ less than $u$, where $m=M(u)$, or $m=\tilde{M}(u)$, or $m=M_{low1}(u)$, depending on whether $(1)$, or $(2)$, or $(3)$ is true, respectively. Thus, we can compute all those $3$-cuts by finding, for every vertex $u$, and every $m\in\{M(v),\tilde{M}(u),M_{low1}(u)\}$, the maximum vertex $v$ in $M^{-1}(m)$ less than $u$, and then check whether $B(u)=B(v)\sqcup\{e\}$, for a back-edge $e$. This last condition is equivalent to $\mathit{b\_count}(u)=\mathit{b\_count}(v)+1$.

\begin{figure*}[t!]
\begin{center}
\centerline{\includegraphics[trim={0 0 0 4.7cm}, clip=true, width=\textwidth]{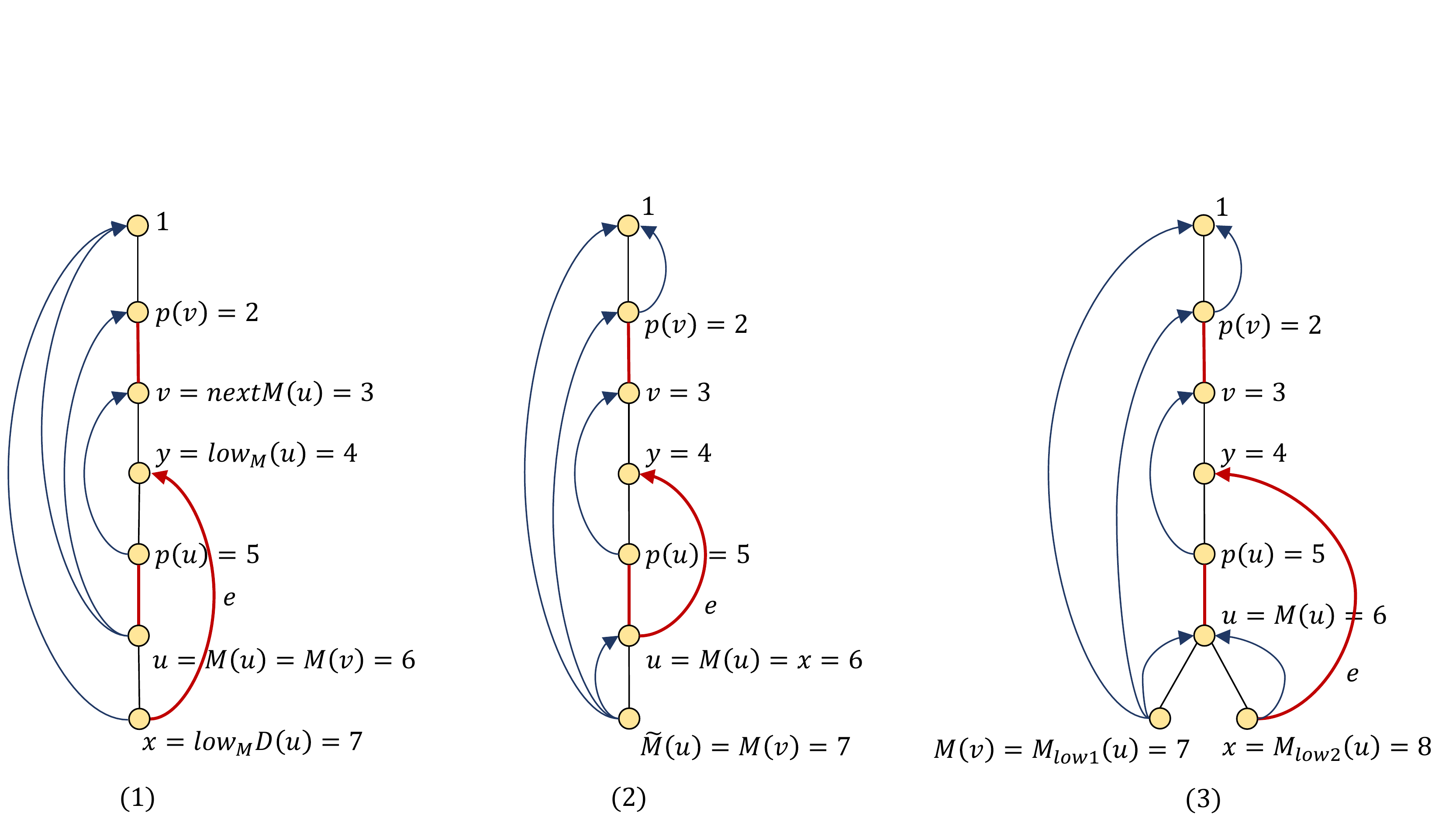}}
\caption{An illustration of the three cases for type-$2$ cuts $\{(u,p(u)),(v,p(v)),e\}$ (shown with red edges) where $v$ is an ancestor of $u$, $e$ is a back-edge, and $B(u)=B(v)\sqcup\{e\}$. (Note that in the examples of Cases (2) and (3) the input graphs have parallel edges.)
Vertices are numbered in DFS order and back-edges are shown directed from descendant to ancestor in the DFS tree.
\label{figure:type2-3cuts-example2}}
\end{center}
\vspace{-1cm}
\end{figure*}

Algorithm~\ref{algorithm:B(u)=B(v)e_nohigh} shows how we can determine all $3$-cuts of this type.
Similar to Algorithm~\ref{algorithm:B(v)=B(u)e_nohigh}, the idea is to handle cases $(1)$, $(2)$ and $(3)$ separately, and our goal is to find, for every vertex $u$, the maximum vertex $v$ in $M^{-1}(m)$ (where $m\in\{M(u),\tilde{M}(u),M_{low1}(u)\}$) which is strictly less than $u$, and then check whether $B(u)=B(v)\sqcup\{e\}$, for a back-edge $e$. We can perform this search in linear time by processing the vertices in a bottom-up fashion, and keep in a variable $\mathit{currentVertex}[m]$ the lowest element of $M^{-1}(m)$ that we accessed so far.
Then we can easily check in constant time whether a pair of vertices $u,v$ such that $v$ is the maximum vertex in $M^{-1}(m)$ which is strictly less than $u$ satisfies $B(u)=B(v)\sqcup\{e\}$, for a back-edge $e$, and also identify this back-edge.

\begin{algorithm}[h!]
\caption{\textsf{Find all $3$-cuts $\{(u,p(u)),(v,p(v)),e)\}$, where $u$ is a descendant of $v$ and $B(u)=B(v)\sqcup\{e\}$, for a back-edge $e$.}}
\label{algorithm:B(u)=B(v)e_nohigh}
\LinesNumbered
\DontPrintSemicolon
initialize an array $\mathit{currentVertex}$ with $n$ entries\;
\textcolor{cblue}{\tcp{\textit{Case (1): $m=M(v)$; just check whether the condition of Lemma~\ref{lemma:find_v_from_u} is satisfied for $\mathit{nextM}(u)$}}}
\ForEach{vertex $u\neq r$}{
  \If{$\mathit{b\_count}(u)=\mathit{b\_count}(\mathit{nextM}(u))+1$}{
    mark the triplet $\{(u,p(u)),(\mathit{nextM}(u),p(\mathit{nextM}(u))),(\mathit{low_{M}D}(u),\mathit{low_M}(u))\}$
  }
}
\textcolor{cblue}{\tcp{\textit{Case (2): $m=\tilde{M}(u)$}}}
\lForEach{vertex $x$}{$\mathit{currentVertex}[x] \leftarrow x$}
\For{$u\leftarrow n$ to $u=1$}{
  $m \leftarrow \tilde{M}(u)$\;
  \lIf{$m=\emptyset$ \textbf{or} $\tilde{M}(u)=M(u)$}{\textbf{continue}}
  \textcolor{cblue}{\tcp{\textit{find the maximum $v\in M^{-1}(m)$ that is smaller than $u$}}}
  $v \leftarrow \mathit{currentVertex}[m]$\;
  \lWhile{$v\neq\emptyset$ \textbf{and} $v\geq u$}{$v \leftarrow \mathit{nextM}(v)$}
  $\mathit{currentVertex}[m] \leftarrow v$\;
  \textcolor{cblue}{\tcp{\textit{check the condition in Lemma~\ref{lemma:find_v_from_u}}}}
  \If{$\mathit{b\_count}(u)=\mathit{b\_count}(v)+1$}{
    mark the triplet $\{(u,p(u)),(v,p(v)),(M(u),l_1(M(u)))\}$
  }
}
\textcolor{cblue}{\tcp{\textit{Case (3): $m=M_{low1}(u)$}}}
\lForEach{vertex $x$}{$\mathit{currentVertex}[x] \leftarrow x$}
\For{$u\leftarrow n$ to $u=1$}{
  $m \leftarrow M_{low1}(u)$\;
  \lIf{$m=\emptyset$ \textbf{or} $l_1(M(u))<u$}{\textbf{continue}}
  \textcolor{cblue}{\tcp{\textit{find the maximum $v\in M^{-1}(m)$ that is smaller than $u$}}}
  $v \leftarrow \mathit{currentVertex}[m]$\;
  \lWhile{$v\neq\emptyset$ \textbf{and} $v\geq u$}{$v \leftarrow \mathit{nextM}(v)$}
  $\mathit{currentVertex}[m] \leftarrow v$\;
  \textcolor{cblue}{\tcp{\textit{check the condition in Lemma~\ref{lemma:find_v_from_u}}}}
  \If{$\mathit{b\_count}(u)=\mathit{b\_count}(v)+1$}{
    mark the triplet $\{(u,p(u)),(v,p(v)),(M_{low2}(u),l_1(M_{low2}(u)))\}$
  }
}
\end{algorithm}

\clearpage

\bibliographystyle{plain}

\end{document}